\documentclass[11pt]{amsart}

\usepackage{amsmath}
\usepackage{amssymb}
\usepackage{algorithm}
\usepackage{algorithmic}
\usepackage{mydefs}
 \usepackage[text={6.5in,9in},centering]{geometry}
\newcommand{\newterm}[1]{\emph{#1}}  
\newcommand{\algo}[1]{\textbf{#1}}   
\newcommand{\maa}{\overline{A}} 

\def\calF{{\mathcal{F}}}
\def\calP{{\mathcal{P}}}

\newcounter{foo}
\newtheorem{theorem}[foo]{Theorem}
\newtheorem{lemma}{Lemma}[section]

\newtheorem{claim}[lemma]{Claim}

\begin{document}

\title[Distributed Scheduling in the SINR Model]{Nearly Optimal Bounds for Distributed Wireless Scheduling in the SINR Model}
\thanks{Supported by grants 90032021 and 120032011 from the Icelandic Research Fund. Preliminary version appeared in ICALP 2011.} 

\author{Magn\'us M. Halld\'orsson}
\address{ICE-TCS, School of Computer Science,
              Reykjavik University, Iceland.}
\email{mmh@ru.is, ppmitra@gmail.com}
\author{Pradipta Mitra}
\date{\today}

\begin{abstract}
We study the wireless scheduling problem in the SINR model. More
specifically, given a set of $n$ links, each a sender-receiver pair,
we wish to partition (or \emph{schedule}) the links into the minimum
number of slots, each satisfying interference constraints allowing
simultaneous transmission.  In the basic problem, all senders transmit
with the same uniform power.

We give a distributed $O(\log n)$-approximation algorithm for the
scheduling problem, matching the best ratio known for centralized
algorithms. It holds in arbitrary metric space and for every
length-monotone and sublinear power assignment. It is based on an
algorithm of Kesselheim and V\"ocking, whose analysis we improve by a
logarithmic factor.  We show that every distributed algorithm uses
$\Omega(\log n)$ slots to schedule certain instances that require only
two slots, which implies that the best possible absolute performance
guarantee is logarithmic.
\end{abstract}

\maketitle

\section{Introduction}
Given a set of $n$ wireless links, each a sender-receiver pair, what is the minimum number of slots needed to schedule
all the links, given interference constraints? This is the canonical problem of \emph{scheduling}
wireless communication, which we study here in a distributed setting.

In a wireless network, simultaneous transmissions on the same channel interfere with each other.
Algorithmic questions for wireless networks depend crucially on the model of interference considered.
In this work, we use the physical, a.k.a.\ \emph{SINR},  model of interference,
precisely defined in Section \ref{sec:model}.
It is known to capture reality more faithfully than the graph-based
models most common in the theory literature, as shown
theoretically as well as
experimentally~\cite{MaheshwariJD08,Moscibroda2006Protocol}. 
Early work on scheduling in the SINR model focused on heuristics
and/or non-algorithmic average-case analysis (e.g.~\cite{kumar00}).  In seminal work, Moscibroda and Wattenhofer \cite{MoWa06}
proposed the problem of scheduling an arbitrary set
of links. Numerous works on various problems in the SINR setting have appeared since.

The \emph{scheduling} problem 
has primarily been studied in a centralized setting. In many realistic scenarios, however,
it is imperative that a distributed solution be found, since a
centralized controller may not exist, and individual nodes in the link
may not be aware of the overall topology of the network. For the
scheduling problem, the only rigorous result previously known is due to
Kesselheim and V\"{o}cking \cite{KV10}, who show that a simple and 
natural distributed algorithm provides an $O(\log^2 n)$-approximation.

In this work, we adopt the algorithm of Kesselheim and V\"{o}cking,
but provide an improved analysis of an $O(\log n)$-approximation. This
matches the best upper bound known for centralized algorithms. Moreover, we show this to be best possible for distributed algorithms that use no external communication infrastructure.


\section{Preliminaries and Contributions}
\label{sec:model}

Given is a set $L = \{l_1, l_2, \ldots, l_n\}$ of links, where
each link $l_v$ represents a communication request from a sender
$s_v$ to a receiver $r_v$. 
The distance between two points $x$ and $y$ is denoted by $d(x,y)$.
The asymmetric distance from link $l_v$ to link $l_w$ is the distance from
$v$'s sender to $w$'s receiver, denoted by $d_{vw} = d(s_v, r_w)$.
Let $\ell_v = d(r_v,s_v)$ denote the length of link $l_v$.


Let $P_v$ denote the power assigned to link $l_v$, or, in other words, $s_v$ transmits with power $P_v$. 
We adopt the SINR model (a.k.a., \emph{physical model})
of interference, in which a node $r_v$
successfully receives a message from a sender $s_v$ if and only if the
following condition holds:
\begin{equation}
 \frac{P_v/\ell_v^\alpha}{\sum_{l_w \in S \setminus  \{l_v\}}
   P_w/d_{wv}^\alpha + N} \ge \beta, 
 \label{eq:sinr}
\end{equation}
where $N$ is a universal constant denoting the ambient noise, 
$\alpha > 0$ denotes the
path loss exponent, $\beta > 0$ denotes the minimum
SINR (signal-to-interference-noise-ratio) required for a message to be successfully received,
and $S$ is the set of concurrently scheduled links in the same \emph{slot}.
%
We say that $S$ is \emph{SINR-feasible} (or simply \emph{feasible}) if (\ref{eq:sinr}) is
satisfied for each link in $S$. 

A power assignment $\calP$ is \newterm{length-monotone} if $P_v \ge P_w$
whenever $\ell_v \ge \ell_w$ and \newterm{sub-linear} if
$\frac{P_v}{\ell_v^{\alpha}} \le \frac{P_w}{\ell_w^{\alpha}}$ whenever
$\ell_v \ge \ell_w$ \cite{KV10}. Two widely used power assignments in
this class are the \emph{uniform} power assignment, where every link
transmits with the same power; and the \emph{linear} power assignment,
where $P_v$ is proportional to $\ell_v^\alpha$. A third one, \emph{mean} power \cite{FKRV09,us:talg12} has also proved to be versatile.

Given a set of links $L$, the \newterm{scheduling problem} is to find
a partition of $L$ of minimum size such that each subset in the
partition is feasible. The size of the partition equals the minimum
number of slots required to schedule all links. We will call this
number the \newterm{scheduling number} of $L$, and denote it by $\chi(L)$
(or $\chi$ when clear from context).

\paragraph{Distributed algorithms.}

A communication infrastructure for running distributed algorithms is
generally assumed to exist in the traditional distributed setting.
The current setting, which abstracts the MAC layer in networks, is
different, as the goal actually is to construct such an
infrastructure. Thus, our algorithm will work with very little global
knowledge and minimal external input. 

Communication is only available over the channel.
Algorithms operate in synchronous rounds with the senders either
transmitting or listening in each round.
When transmission is successful, the sender stops transmitting.
This necessitates an acknowledgment
from the receiver, so that the sender knows when his message has been
heard. These acknowledgments are sent over the same channel as the
message; thus, there are no side-channels for control messages. 
We shall assume this model, which we call \emph{ack-only}, in the rest of the paper.

We assume that nodes have a rough estimate of the network size $n$ 
and (senders of) links are assigned a fixed
length-monotone, sublinear power function.
The power assignment indirectly requires knowledge of distances and
the path loss constant $\alpha$ and the technological parameters
$\beta$ and $N$. No information of locations is
needed.

We note that the assumptions are particularly minimal when
using uniform power.  The algorithm then needs no knowledge of
distances, the path loss constant $\alpha$, nor the technological
parameters $\beta$ and $N$.  Only the polynomial bound on the number
$n$ of nodes is needed.

\paragraph{Affectance.}
We will use the notion of \emph{affectance}, introduced in
\cite{GHWW09,HW09} and refined in \cite{KV10} to the thresholded form
used here.
The affectance $a^\calP_w(v)$ \emph{on} link $l_v$ \emph{from} another link $l_w$,
with a given power assignment $\calP$,
is the interference of $l_w$ on $l_v$ relative to the power
received, or
  \[ a^\calP_{w}(v) 
     = \min\left\{1, c_v \frac{P_w/d_{wv}^\alpha}{P_v/\ell_v^\alpha}\right\}\ ,
  \] 
where $c_v = \beta/(1 - \beta N \ell_v^\alpha/P_v)$ 
depends only on model constants and on the length of $l_v$.
We will drop $\calP$ and assume it to be an arbitrary length-monotone
sub-linear power strategy, unless otherwise stated.
Let $a_v(v) = 0$.
For a set $S$ of links and a link $l_v$, 
let $a_S(v) = \sum_{l_w \in S} a_w(v)$, referred to as \emph{in-affectance},
and $a_v(S) = \sum_{l_w \in S} a_v(w)$, the \emph{out-affectance} from $l_v$.
For sets $S$ and $R$, $a_R(S) = \sum_{l_v \in R}\sum_{l_u \in S} a_v(u)$.
Using such notation, (\ref{eq:sinr}) can be rewritten as 
\begin{equation}
a_S(v) \leq 1 \ ,
\label{eq:feasaff}
\end{equation}
whenever $|S| > 2$, and this is the form we will use.

\subsection{Related Work}

In the centralized setting, scheduling results have closely followed
results on the related \emph{capacity} problem, where one wants to
find the maximum subset of $L$ that can be transmitted in a single
slot). Goussevskaia et al.~\cite{gouss2007} showed the problem to be 
NP-hard for the case of uniform power on the plane and gave $O(\log \Delta)$-approximation result (on the plane), where $\Delta$
denotes the ratio between the maximum and minimum length of a link.
Same bound was shown by Andrews and Dinitz \cite{AndrewsD09} but in comparison with optimum that is allowed to choose arbitrary power.
Constant factor approximation was obtained for uniform power, also on the plane, by Goussevskaia et al.~\cite{GHWW09}, which was generalized to all length-monotone, sublinear power assignments and arbitrary metrics space by Halld\'orsson and Mitra \cite{SODA11}. Kesselheim \cite{KesselheimSoda11} gave a constant-factor approximation for the joint problem of selecting links and assigning them feasible power (see also earlier work of Chafekar et al.~\cite{chafekar07}.

All the results lead to equivalent bounds for the centralized scheduling problem with $O(\log n)$-factor overhead. In particular, $O(\log n)$-approximation holds for scheduling with length-monotone, sublinear power \cite{SODA11} and with arbitrary power control \cite{KesselheimSoda11}.
Also, the problem remains NP-hard \cite{gouss2007}.
For the results in terms of $\Delta$ on the plane \cite{gouss2007,AndrewsD09}, 
this overhead can be avoided (see, e.g., \cite{us:talg12}).
Scheduling with arbitrary power control can also be approximated within a factor of $O(\log n \log\log \Delta)$ when the algorithm uses mean power.
For linear power on the plane, an algorithm using 
$O(\chi + \log^2n)$ slots for instances with optimal schedule length $\chi$
was given by Fangh\"anel et al.~\cite{FKV09}; on the plane, this can be improved to a constant factor \cite{tonoyanlinear}. 
A bi-directional version was studied by Fangh\"{a}nel et al. \cite{FKRV09} 
and further treated in \cite{us:talg12,SODA11} and 
the joint multi-hop scheduling and routing was treated by Chafekar et al.~\cite{chafekar07}.

In the distributed setting, the capacity problem
was treated with no-regret learning by Dinitz \cite{Dinitz2010}
culminating in a $O(1)$-approximation algorithm for uniform power
of \'Asgeirsson and Mitra \cite{GT2011}. 
However, these game-theoretic algorithms take time
polynomial in $n$ to converge, and thus can be viewed more appropriately as
determining capacity instead of realizing it in ``real time''.

For distributed scheduling, the only work that we are aware of is the
groundbreaking paper of Kesselheim and V\"ocking \cite{KV10}, who give
a distributed $O(\log^2 n)$-approximation algorithm for the scheduling
problem with fixed length-monotone and sublinear power assignment.
Our results constitute a $\Omega(\log n)$-factor
improvement. Kesselheim and V\"ocking also extend their results to
multi-hop scheduling, with the same approximation factor, for which
our improvements do not apply, and to routing, with an extra logarithmic factor.

A versatile measure introduced in \cite{KV10} is 
the maximum average affectance $\maa$ of a link set $L$, 
defined as
$$\maa(L) := \max_{R \subseteq L} \text{avg}_{l \in R} a_R(l)
= \max_{R \subseteq L} \frac{a_R(R)}{|R|}\ .$$
They then show two results that combined yield the $O(\log^2
n)$-approximation factor. On the one hand, they show that $\maa(L) =
O(\chi(L) \log n)$. On the other hand, they present a natural algorithm
(which we also use in this work) that schedules links in $O(\maa(L)
\log n)$ slots. We show that both of these bounds are tight. Thus, it
is not possible to obtain improved approximation using the measure $\maa$.


Following the original publication of this work, the results have been
applied to distributed connectivity and aggregation
\cite{PODC12,PODC13}.  A different approach for distributed capacity
was proposed by Pei and Kumar \cite{pei2012}, with complexity that is
a function of the link lengths.  In a recent follow-up work,
Halld\'orsson et al.~\cite{us:SODA13} have shown that $\maa(L) = O(\chi)$
for all sublinear, length-monotone power assignments other than
uniform power.

\subsection{Our Contributions}
We achieve the following results:

\begin{theorem}
There is a $O(\log n)$-approximate distributed algorithm for the
scheduling problem, in arbitrary metric space and for all length-monotone
sublinear power assignments.
\end{theorem}

\begin{theorem}
For every $n$, there is an instance $L_n$ of links on the real line
that can be scheduled in two slots but for which every
every distributed algorithm uses $\Omega(\log n)$-slots (w.h.p).
Thus, $\Theta(\log n)$ is the best absolute approximation factor for
a distributed scheduling algorithm.
\end{theorem}

As in \cite{KV10}, our upper bound results hold in arbitrary distance metrics (and do not require the common assumption that $\alpha > 2$). We also show that the results hold independent of the ambient noise term $N$, extending \cite{KV10}.
The lower bound result necessarily holds independent of power assignment strategy and for all settings of the technological constants $\alpha$, $\beta$ and $N$.

One of our main technical insights is to devise a different measure that 
involves \emph{median} rather than average affectance.
The measure $\Lambda = \Lambda(L)$ is given by
\[ \Lambda(L) := \max_{R \subseteq L} \text{median}(A(R)) \ , \]
where $A(R) = \{a_R(l) : l \in R\}$ is the multi-set of
in-affectance values of links in the subset $R$,
and $\text{median}(X)$ denotes the median of a multi-set $X$.
%
Since we only insist that half of the given subset $R$ of
links have affectance bounded by $\Lambda$, 
the value of $\Lambda$ may be much smaller than $\maa$. Indeed,
we show that $\Lambda = O(\chi)$ and that the algorithm schedules all
links in time $O(\Lambda \log n)$, achieving the claimed approximation
factor. 

The other main technical contribution of the paper is the
introduction of the concept of \emph{anti-feasibility}. A set $S$ of
links is anti-feasible \footnote{For a technical reason we use a different constant here than for feasibility; the signal-strengthening result of \cite{HW09}  implies that this only affects constants in the approximation factors.} if $a_v(S) \le 2$, for every $l_v$ in $S$;
i.e., if the out-going affectance from each link is small. A set is
\emph{bi-feasible} if it is both feasible and anti-feasible. We
observe in this paper that every feasible set contains a large
bi-feasible set and that certain analyses are easier on bi-feasible
sets. This has proved useful in later works, e.g., in giving
simplified analysis of capacity approximation algorithms
\cite{KesselheimThesis,us:INFOCOM12}.

In the next section, we give the improved analysis of a $O(\log
n)$-factor for distributed scheduling, via the measure $\Lambda$; the
treatment of acknowledgments is given in Section \ref{sec:acks}. We
show in Section \ref{sec:lbound} that this logarithmic factor is best
possible, and give a construction in Section \ref{sec:bara} that
shows that this result cannot be obtained in terms of the
measure $\maa$.

\section{$O(\log n)$-Approximate Distributed Scheduling Algorithm}

The algorithm from \cite{KV10}, listed below as \algo{Distributed},
is a natural backoff scheme, in the tradition of  ALOHA \cite{Tanenbaum:2002:CN:572404}.
It is run synchronously, but independently, on each sender of a link. 
The algorithm, and all the results in this section, work for an
arbitrary fixed sublinear length-monotone power assignment.

\begin{algorithm}                      
\caption{Distributed}          
\label{alg3}                           
\begin{algorithmic}[1]                    
     \STATE $k \leftarrow 0$
     \LOOP \label{needack}
       \STATE $q = \frac{1}{4 \cdot 2^k}$ \label{setq}
       \FOR{$\frac{4}{q} c_1 \ln n$ slots}
          \STATE transmit with i.i.d.\ probability $q$
          \IF{successful (and acknowledged)}
            \STATE halt
          \ENDIF
       \ENDFOR
       \STATE $k \leftarrow k + 1$
     \ENDLOOP
\end{algorithmic}
\label{alg13ig}
\end{algorithm}

The algorithm is mostly self-descriptive. The constant $c_1$ is to be chosen to satisfy the high probability bound desired.
One point to note is that Line \ref{needack} necessitates some sort of
acknowledgment mechanism for the distributed algorithm to stop. For simplicity, we will defer the issue
of acknowledgments to Section \ref{sec:acks} and simply assume their existence for now. Thm.~\ref{thm:mainth} below implies our main positive result. 
Let $\Lambda = \Lambda(L)$.

\begin{theorem}
If all links of a set $L$ of $n$ links run \algo{Distributed}, 
then $L$ is fully scheduled in $O(\Lambda \log n)$ slots, with high probability.
\label{thm:mainth}
\end{theorem}

To prove Thm.~\ref{thm:mainth}, we claim the following. 
\begin{lemma}
Consider a subset $R \subseteq L$ of links and a particular time slot $t$ in which each sender of $R$ transmits with probability $q \leq \frac{1}{2 \Lambda}$. 
Then, the expected number of successful transmissions is at least $\frac{q \cdot |R|}{4}$.
\label{mainlem1}
\end{lemma}
\begin{proof}
Define $M = M_{\Lambda}(R) = \{l_u \in R : a_{R}(u) \le \Lambda\}$. 
By the definition of $\Lambda$, $|M| \ge |R|/2$.
Thus, it suffices then to show that at least $q|M|/2$ transmissions in slot $t$ are successful in expectation. 

Intuitively, the success probability of a link is proportional to its
in-affectance. The links in $M$ are the ones with low in-affectance, so as long as the transmission probability $q$ is less than $1/(2\Lambda)$, they will succeed with probability $1/2$ if transmitting.

For $l_u \in R$, recall that $T_{u} = T_u(t)$ is the indicator random variable that
link $l_u$ transmits, and let $S_{u}=S_u(t)$ be the indicator random variable that $l_u$ succeeds.

We shall make use of a few elementary facts about probabilities.
For a (Bernoulli) indicator random variable $X$, $\Ex(X) = \Pr(X)$.
For random variables $X_1, X_2, \ldots$, it holds by the 
linearity of expectation that $\sum_i \Ex(X_i) = \Ex(\sum_i X_i)$.
And, for a random variable $X$ that assumes non-negative values, 
$\Pro(X > 1) \le \Ex(X)$.

Armed with these facts, we can now bound
the probability that a transmitting link $l_u \in M$ is unsuccessful:
\begin{align*}
\Pro(S_u = 0|T_u = 1)  & = \Pro\left(\sum_{l_v \in R} a_{v}(u) T_{v} > 1\right) \\
& \leq \Ex\left(\sum_{l_v \in R} a_{v}(u) T_{v}\right)  \\
& = \sum_{l_v \in R} a_{v}(u) \Ex(T_{v})  \\
& =  q \sum_{l_v \in R} a_{v}(u) \leq q \cdot  \Lambda \ ,
\end{align*}
where the first equality uses (\ref{eq:feasaff}),
and the last inequality uses the definition of $M$. 

Thus, when $q \le \frac{1}{2\Lambda}$,
\[ \Pro(S_u = 0|T_u = 1) \le 1/2 \ , \]
which allows us to bound the 
probability of link $l_u$ transmitting in the time slot by
\begin{align*}
\Ex(S_{u}) & = \Pro(S_{u} = 1) \\
& = \Pro(T_{u} = 1)\Pro(S_{u} = 1|T_u = 1) \\
& = q (1 - \Pro(S_u = 0|T_u = 1)) \\
& \ge q/2\ .
\end{align*}
The expected number of successful links in the time slot
is then
\begin{align*}
\Ex\left(\sum_{l_u \in R}S_u\right) & = \sum_{l_u \in R} \Ex(S_u)  \\
  & \geq \sum_{l_u \in M} \Ex(S_u)   \ge  |M| \cdot q/2 \ge |R| \cdot q/4 , 
\end{align*}
implying the lemma.
\end{proof}


\begin{proof}[of Thm.~\ref{thm:mainth}]
Given Lemma \ref{mainlem1}, the theorem 
follows essentially from the arguments in Thms.\ 2 and 3 of \cite{KV10}.
Let $\hat{q} = 2^{-(1+\lceil \lg \Lambda\rceil)}$, i.e., the unique power of two 
satisfying $\frac{1}{4\Lambda} \le \hat{q} \le \frac{1}{2\Lambda}$.

We first bound the probability that not all links
are scheduled during the iteration of the outer loop 
when $q$ in Line \ref{setq} equals $\hat{q}$.

Let $\hat{t}$ be the first time slot where $q \le \hat{q}$.
Let $n_t$ be the random variable indicating the number of links that did not successfully transmit in the first $t$ time slots.

Lemma \ref{mainlem1} implies that for any given value $s$ and time slot $t \ge \hat{t}$,
  \[ \Ex(n_{t}|n_{t-1}=s) \leq s - \frac{\hat{q}}{4}s\ , \] 
and thus
$$\Ex(n_{t}) \leq \sum_{s=0}^{\infty} \Pro(n_{t-1} = s) \cdot (1 - \hat{q}/4) s 
  = (1 - \hat{q}/4) \Ex(n_{t-1})\ .$$
Noting that $n_0 = n$, this yields that 
\[ \Ex(n_t) \leq (1 - \hat{q}/4)^t n\ . \] 
Now, after $\hat{t} + 4 c_1 \ln n/\hat{q}$ time slots,
the expected number of remaining requests is 
\begin{align*}
\Ex(n_{\hat{t}+ 4 c_1 \ln n/\hat{q}}) & \leq (1 - \hat{q}/4)^{4 c_1 \ln n/\hat{q}} \Ex(n_{\hat{t}}) \\
  & \leq \left(\frac{1}{e}\right)^{c_1 \ln n} n = n^{1 - c_1} \ .
\end{align*}
By Markov's inequality, 
\begin{align*}
\Pro(n_{\hat{t}+ 4 c_1 \ln n/\hat{q}} \neq 0) & = \Pro(\hat{t}+ n_{4 c_1 \ln n/\hat{q}} \geq 1) \\
  & \leq \Ex(n_{\hat{t}+ 4 c_1 \ln n/\hat{q}}) \leq n^{1-c_1}\ . 
\end{align*}
Thus, with high probability all the links are scheduled while $q \ge \hat{q}$.

Finally, to bound the total running time of the algorithm, we sum up the
spent for values of $q$ smaller than $\hat{q}$, bounding $t_0$.
This is a geometric series given by
\begin{align*}
t_0 = & \sum_{i=2}^{\lg (1/\hat{q})} \frac{8 c_1 \ln n}{2^{-i}}  \\
 & = 8 c_1 \ln n \sum_{i=2}^{\lg (1/\hat{q})} 2^i \\
 & \le 8 c_1 \ln n \cdot 2^{\lg(1/\hat{q})+1} \\
 & = 8 c_1 \ln n \cdot \frac{2}{\hat{q}} \\
 & \le 64 c_1 \Lambda \ln n \ ,
\end{align*}
establishing the time complexity.
%
\end{proof}

\subsection{Bounding the Measure}

We need the following lemma to get a handle on affectances.
Recall that we assumed that the implicit power assignment is length-monotone and sublinear. 
\begin{lemma}[Lemma 7, \cite{KV10}]
Let $L$ be a feasible set and $l_u \not\in L$ be link with $\ell_u
\leq \ell_v$ for all $l_v \in L$. Then,
$a_{L}(u) = O(1)$.

\label{incoming}
\end{lemma}

We now prove the following complementary result. It can be contrasted with Lemma 9 of \cite{KV10}, which without the anti-feasibility condition can only
give $a_{v}(L) = O(\log n)$.
The second part of the lemma essentially follows Lemma 11 of \cite{GT2011} (which had the unnecessary assumption that $L$ is feasible).


We first need the following result.
\begin{lemma}[\cite{us:talg12}]
Let $l_u, l_v$ be links with $\min(a_u(v),a_v(u))\le{}1/q$.
Then, $d_{uv} \cdot d_{vu} \ge q^2 \cdot \ell_u \ell_v$. 
\label{lem:ind-separation}
\end{lemma}

\begin{lemma}
Let $L$ be an anti-feasible set with length-monotone and sublinear power
and let $l_v \not\in L$ be a link with
$\ell_v \le \ell_u$, for every $l_u \in L$.
Then, $a_{v}(L) = O(1)$.
\label{outgoing}
\end{lemma}
\begin{proof}
We first use a variation of the signal strengthening technique of \cite{HW09},
given as Thm.~\ref{anti-signal-strengthening} in the Appendix.
This allows us to decompose
the set $L$ into $\lceil 4 \cdot 3^\alpha \rceil^2$ sets, 
where each set $S$ satisfies $a_w(S) \leq \frac1{3^{\alpha}}$, for all
$l_w \in S$. 
We shall prove the claim for $S$; the claim will then hold for $L$ by summing over the $\lceil 4 \cdot 3^\alpha \rceil^2$ sets.

Let $l_u = (s_u, r_u)$ ($l_w = (s_w, r_w)$) be the link in $S$ whose
sender (receiver) is closest to $s_v$, i.e., $d(s_v,s_u) \le \min_{l_x \in S} d(s_v,s_x)$ ($d(s_v,r_w) \le \min_{l_x \in S} d(s_v,r_x)$), respectively.
Let $h = d(s_v, s_{u})$. 
We claim that for all links $l_x$ in $S$, $l_x \neq l_w$, it holds that
\begin{equation}
d(s_v, r_{x}) \geq \frac{1}{2} h \ .
\label{eqn:dist1}
\end{equation}

To prove this, assume, for contradiction, that $d(s_v, r_{x}) < \frac{1}{2} h$. 
Then, by the definition of $l_w$, $d(s_v, r_{w}) < \frac{1}{2} h$,
and by the definition of $l_u$, $d(s_v, s_x) \ge d(s_v,s_u) \geq h$ 
and $d(s_v, s_w) \geq h$. 
Thus, $\ell_w \geq d(s_v, s_w) - d(s_v, r_w) > \frac{h}{2}$ 
  and similarly $\ell_{x} > \frac{h}{2}$. 
On the other hand, by the triangular inequality and the assumed inequality,
\[ d(r_{w}, r_{x}) \le d(r_w,s_v) + d(s_v,r_x)
     < \frac{h}{2} + \frac{h}{2} < h \ . \]
Now, 
\begin{align*}
d_{w x} \cdot d_{x w} & \leq (\ell_{w} + d(r_{w}, r_{x}))(\ell_{x} +d(r_w, r_{x}))  \\
  & < (\ell_{w} + h)(\ell_{x} + h) \\ 
  & < 9 \ell_{w} \ell_x\ ,
\end{align*}
contradicting Lemma \ref{lem:ind-separation}.  This establishes
(\ref{eqn:dist1}).

Now, by the triangular inequality, the definition of $h$ and (\ref{eqn:dist1}),
\[ d_{u x}  = d(s_{u}, r_{x}) \leq d(s_{u}, s_v) + d(s_v, r_{x}) \leq 3 d(s_v, r_{x}) = 3 d_{v x}\ . \] 
We observe that $P_v \le P_u$ holds by length-monotonicity.
Also, note that since the maximum affectance between links in $S$ is 
$\frac1{3^{\alpha}}$, the thresholding in the affectance definition does not take effect, implying that 
$a_{u}(x) = c_x\frac{ P_{u}}{d_{u x}^{\alpha}} \frac{\ell_x^{\alpha}}{P_{x}}$.
Thus, 
\[
a_{v}(x) = c_x \frac{P_{v}}{d_{vx}^{\alpha}} \frac{\ell_x^{\alpha}}{P_{x}} 
   \leq c_x \frac{3^{\alpha} P_{u}}{d_{u x}^{\alpha}} \frac{\ell_x^{\alpha}}{P_{x}} 
   = 3^{\alpha} a_{u}(x)\ .  \]
Finally, summing over all links in $S$, 
\begin{align*}
a_v(L) & = a_{v}(w) + \sum_{l_x \in S \setminus \{l_w\}}  a_{v}(x)\\
  & \leq 1 + 3^\alpha \sum_{l_x \in S \setminus \{l_w\}}  a_{u}(x) \\
  & \leq 1 + 3^\alpha \cdot 2 = O(1) \ , 
\end{align*}
using anti-feasibility in the last inequality.
The lemma follows.
\end{proof}


We can now derive the needed bound on the measure.
\begin{theorem}
Let $L$ be a set of links.
Then, $\Lambda(L) = O(\chi(L))$.
\label{mainlem2}
\end{theorem}
\begin{proof}
Let $\chi = \chi(L)$ and let $R$ be an arbitrary subset $R \subseteq L$.
To prove the theorem, it suffices to show that 
at least half of the links in $R$ have in-affectance at $O(\chi(L))$.



Consider a partition of $R$ into $\chi$ feasible subsets $S_1, S_2, \ldots,
S_\chi$, and define $S'_i = \{l_v \in S_i : a_v(S_i) \leq 3\}$.
We claim that $S'_i$ contains at least two thirds of the links in $S_i$.
\begin{claim}
For all $i$, $|S'_i| \geq \frac{2|S_i|}{3}$.
\end{claim}
\begin{proof}
Since $S_i$ is feasible, it follows from (\ref{eq:feasaff}) that
$a_{S_i}(v) \leq 1$, for every link $l_v \in S_i$. 
Let $\hat S_i = S_i \setminus S'_i$.
Now, 
\[ a_{\hat S_i}(S_i) \le \sum_{l_v \in S_i}a_{S_i}(v) \le \sum_{l_v \in  S_i} 1
   \leq |S_i|\ . \]
But, $a_{\hat S_i}(S_i) = \sum_{l_v \in \hat S_i}a_v(S_i) \geq 3 \cdot |\hat S_i|$, 
by the definition of $\hat{S_i}$. Thus, $|\hat S_i| \leq 2|S_i|/3$, proving the claim.
\end{proof}

\noindent Let $R' = \cup_i S'_i$. By the above claim, $3|R'|/4 \geq |R|/2$.

We next show the following.
Let $c_2$ ($c_3$) be the constant implicit in the big-oh notation in
Lemma \ref{incoming} (Lemma \ref{outgoing}), respectively.

\begin{claim}
 $a_R(R') \le (c_2 + c_3) |R| \cdot \chi.$
\end{claim}

\begin{proof}
We first observe that for every $i, j$,
\begin{align}
a_{S_j}(S'_i) & = \sum_{l_u \in S'_i} \sum_{l_v \in S_j} a_{v}(u) \quad\quad & \nonumber \\
& \leq \sum_{l_u \in S'_i}  \sum_{\substack{l_v \in S_j \\ \ell_v \geq \ell_u  }} a_{v}(u)  
  + \sum_{l_u \in S'_i}  \sum_{\substack{l_v \in S_j \\ \ell_v \le \ell_u  }} a_{v}(u)  
    \nonumber \\
& \le  \sum_{l_u \in S'_i} c_2 + 
\sum_{l_v \in S_j} \sum_{\substack{l_u \in S'_i \\ \ell_u \geq \ell_v}} a_{v}(u) 
\nonumber \\
& \le  c_2|S_i'| + \sum_{l_v \in S_j} c_3 
\nonumber \\
& \le c_2 |S_i| + c_3|S_j| \ , \label{xbound}
\end{align}
using Lemma \ref{incoming} and rearrangement in the second inequality, and Lemma \ref{outgoing} in the third inequality.
We then obtain that
\begin{align*}
 a_R(R') & =  \sum_{i = 1}^\chi  \sum_{j = 1}^\chi a_{S_j}(S'_i) \\
   &  \le \sum_{i,j=1}^\chi c_2 |S_i| + c_3 |S_j| && \text{(By (\ref{xbound}))} \\
   &  = \sum_{i,j=1}^\chi (c_2 + c_3) |S_i|) && \text{(By symmetry)}\\
   &  = (c_2 + c_3) \chi \sum_{i=1}^\chi|S_i| \\
   &  = (c_2 + c_3) \chi |R| && \text{(Defn.~of $S_i$)}
\end{align*}
\end{proof}

 It follows that the average in-affectance $a_{R}(l')$ over the links $l' \in R'$ is 
 at most 
\[ \frac{a_R(R')}{|R'|} \le \frac{(c_3+c_4)|R|\cdot \chi}{|R'|} 
   \le \mu := \frac{3 (c_2 + c_3)}{2} \chi\ .  \]
Recall that $M_{4\mu}(R) = \{l \in R : a_R(l) \le 4 \mu \}$ is the set of links in $R$ of in-affectance at most four times the average.
By Markov's inequality, at least three fourths
 of the links have in-affectance at most four times the average;
namely, 
\[ |M_{4\mu}(R)| \ge |M_{4\mu}(R')|\ge 3|R'|/4 \ge |R|/2 \ . \]
That is, at least half the links in $R$ have in-affectance at most $4\mu$.
Hence, the median in-affectance of links in $R$ is bounded above by 
\[ \text{median}(A(R)) \le 4\mu = O(\chi)\ . \]
Since this holds for every given $R$, the theorem follows.

\end{proof}

\subsection{Acknowledgments}
\label{sec:acks}

In the preceding exposition, we ignored the issue of sending acknowledgments from receivers to senders. 
We can treat acknowledgments in a fashion similar to Kesselheim and
V\"ocking \cite{KV10}.
We outline their approach briefly, but direct the reader to their
paper for the details.

A special slot for acknowledgments is inserted between the time slots
used by Algorithm \ref{alg13ig}.  A node that successfully received a
packet will transmit an acknowledgment with probability $p = 1/8$.
The power $P^*_v$ used for the acknowledgment on link $l_v$ is
chosen to be proportional to $P^*_v = \ell^\alpha/P_v$ (using the
right scaling factor).
Kesselheim and V\"ocking show that at least half of these
acknowledgments are successful in expectation.
That implies that we can modify Lemma \ref{mainlem1} to claim that
the expected number of successfully acknowledged transmissions is at
least $p \cdot q |R|/4 = q|R|/32$, losing only a constant factor. 
The rest of the arguments are then identical.

The only catch is that they do assume in their analysis that there are
no weak links in the instance; a link $l_v$ is said to be \emph{weak}
iff $c_v > C \beta$, for an appropriately chosen constant $C$ (whose value affects the choice of $p$). 
We show here how to extend the approach to deal
with weak links. For simplicity of exposition, we illustrate it for the case of uniform power and assume that weak links satisfy $c_v > 3 \max(\beta, 1)$.

The original transmissions, using Algorithm \ref{alg13ig}, are
unchanged, but we allocate a separate time slot for the
acknowledgments of weak links.  Each receiver of a successfully
transmitting weak link sends an acknowledgment in that time
slot with probability $p'$ (to be chosen).

The key observation in the following lemma is that weak links must be spatially well-separated. This implies that differences between the positions of the sender and receiver of a link are minor,
allowing us to relate the success probability for an acknowledgment in terms of the observed success of the original transmission.

\begin{lemma}
Assume the use of uniform power.
Let $l_v$ be a weak link that transmits successfully in a given time slot $t$ of Algorithm \ref{alg13ig}. Then, the transmission is successfully acknowledged 
with probability $p'/2$ when $p' \le (1+2\ln 3 \cdot \alpha)^{-\alpha}$.
\end{lemma}

\begin{proof}

Let $l_u$ be another weak link that successfully transmitted at time $t$.
Since both were successful, (\ref{eq:feasaff}) is satisfied in both directions, which implies that
\begin{equation}
 d_{uv}^\alpha \ge c_v \ell_v^\alpha \ge 3 \ell_v^\alpha, ~\text{and}~ 
   d_{vu}^\alpha \ge 3 \ell_u^\alpha \ .
\label{eq:weak1}
\end{equation}
By the triangular inequality, $d_{uv} \le d_{vu}+\ell_u + \ell_v$,
which by (\ref{eq:weak1}) implies that
\[ d_{uv} \left(1 - \frac{1}{3^{1/\alpha}} \right) \le d_{uv} - \ell_v
  \le d_{vu} + \ell_u \le \left(1 + \frac{1}{3^{1/\alpha}} \right)d_{vu} \ .  \]
Now, observe that
\[ \frac{\left(1 + \frac{1}{3^{1/\alpha}}\right)}{\left(1 - \frac{1}{3^{1/\alpha}}\right)} \le 1 + \frac{2}{e^{(\ln 3)/\alpha}-1} \le 1 + \frac{2\alpha}{\ln 3}\ . \]
Thus, 
\begin{equation}
d_{uv}^\alpha \le \left(1 + \frac{2\alpha}{\ln 3}\right)^\alpha d_{vu}^\alpha\ .
\label{eq:inv-bnd}
\end{equation}

Now, let $l^*_v = (s^*_v, r^*_v) = (r_v,s_v)$ be the \emph{dual} link of $l_v$, with the roles of sender and receiver reversed. A transmission on $l_v$ is acknowledged on $l^*_v$.
We use (\ref{eq:inv-bnd}) to bound the in-affectances of a dual link $l^*_v$
from another dual link $l^*_u$:
\begin{align*}
a_{u^*}(v^*) & = c_v \left(\frac{\ell_v}{d_{vu}}\right)^\alpha \\
            & \le \left(1 + \frac{2\alpha}{\ln 3}\right)^\alpha 
                c_v \left(\frac{\ell_v}{d_{uv}}\right)^\alpha \\
            & = \left(1 + \frac{2\alpha}{\ln 3}\right)^\alpha a_u(v) \ .
\end{align*}
Let $S$ be the set of weak links that successfully transmitted in slot $t$
and $S^*$ the set of the corresponding dual links.
Suppose each link in $S$ transmits an acknowledgment with probability $p'$.
Then, the expected in-affectance of a link $l_v^*$ that transmits an ack is
bounded by
\begin{align*}
\Ex\left(\sum_{l_u^* \in S} p' \cdot a_{u^*}(v^*)\right)
 & \le \sum_{l_u^* \in S^*} \Ex(p' a_{u^*}(v^*)) \\
 & = p' \sum_{l_u^* \in S^*} a_{u^*}(v^*) \\
 & \le p' \left(1 + \frac{2\alpha}{\ln 3}\right)^\alpha \sum_{l_u^* \in S^*} a_u(v) \\
 & \le \frac{1}{2} a_S(v) 
   \le \frac{1}{2}\ , \\
\end{align*}
using the feasibility of $S$.
Hence, the probability that a link receives less than twice the expected in-affectance is at least $1/2$, i.e.,
a dual link that does attempt to transmit an acknowledgment
has at least $50\%$ chance of success.
The probability that a given link both attempts to send an acknowledgment and that the transmission is successful, is then at least $p'/2$.
\end{proof}

\section{$\Omega(\log n)$-Factor Lower Bound for Distributed Scheduling}
\label{sec:lbound}
We construct a set of $2 n$ unit length links on the line that can be scheduled in two slots while no distributed algorithm can 
schedule the set in less than $\Omega(\log n)$ slots.

We assume 
that all senders start at the same time in the same state and use the same (randomized) algorithm.
Note that the algorithm presented operates under these assumptions.

For simplicity, we assume the noise $N = 0$, but note that the
construction can be modified to hold for different values of $N$.  We
allow $\alpha$ and $\beta$ to be arbitrary positive values.  We start
with a \emph{gadget} $F$ with two identical links of length $1$, in a
yin-yang position, i.e., with the sender of one link in the same
position as the receiver of the other (it suffices that they be separated
by at most $(P_{max}/(\beta P_{min}))^{1/\alpha}$, where $P_{max}$ ($P_{min}$) is the maximum (minimum) power that can be used, respectively).
Let $x = (2\beta
n)^{1/\alpha}$.  The construction consists of $n$ such gadgets $F_i, i
= 1, 2, \ldots n$, placed on the line as follows: The sender of one
link and the receiver of the other link in $F_i$ are placed at point
$i (x+1)$ and the other two nodes of $F_i$ are placed at $i (x+1) +
1$.

The construction ensures that \emph{a link successfully transmits only
  if the other link in the gadget does not transmit}. This holds
independent of the power used on these links.
On the other hand, when using uniform power, 
the affectance from links of other gadgets is negligible.
To see this, consider the affectance on a link $l_u \in F_i$ from
all links of other gadgets, i.e., from all links $l_v \in \hat F :=
\cup_{j \neq i}F_j$. There are $2n - 2$ links in $\hat F$. The
distance $d_{vu} \geq x$. Therefore, $\sum_{\hat F}a_{v}(u) \leq (2n-2)
\frac{\beta}{x^{\alpha}}  < 1$. 
Thus, behavior of links in other gadgets is immaterial to the success
of a link.  This also implies that the scheduling number of this set
of links is $2$.  Note that since the construction uses equi-length
links, the only possible oblivious power assignment is the uniform
one.

To prove the lower bound, we say that gadget $F_i$ is \emph{active}
at time $t$ if neither link of $F_i$ has succeeded by time $t-1$,
and denote the event by $A_i(t)$.
Let $T_{u}(t)$ denote the indicator random variable that link $l_u$ 
transmits at time $t$.

\begin{lemma}
Let $F_i$ be a gadget and $t\ge 0$ be a time.
The transmission probabilities of the two links in $F_i$ at time $t$
are identical and independent, conditioned on $F_i$ being active
at time $t$.
\label{lbiid}
\end{lemma}
\begin{proof}
Let $l_u$ and $l_v$ be the links in gadget $F_i$.
Let $T_u=T_u(t)$ and $T_v=T_v(t)$, for short. By symmetry,
the distributions of $T_u$ and $T_v$ are identical, thus we need only to prove
their independence.

We can model the randomness used by the algorithms as an i.i.d.\ random
choice over a set $\calF$ of functions. Each $f \in \calF$ is a function 
 that takes a history of past transmissions and receptions over previous slots, and returns a binary transmission decision. 
Note that if $A_i(t)$ occurs then the histories of $l_u$ and $l_v$
over the previous $t-1$ slots are identical. 
The different histories that can result in $A_i(t)$ occurring are
disjoint; thus, it is enough to prove independence for a fixed history $H$.
Let $f_u$ and $f_v$ denote the
functions chosen by $l_u$ and $l_v$, and allow them also to represent the event that they get chosen. Once again, by symmetry,
there is some $\calF' \subseteq \calF$ such that $H$ happens iff 
$f_u \in \calF'$ and $f_v \in \calF'$. We will use the Iverson bracket
$[X]$ to denote the value $1$ if $X$ is true and $0$ otherwise.

\noindent Then, for fixed Boolean outcomes $a$ and $b$, 
\begin{align*}
\lefteqn{\Pro(T_u = a, T_v = b \,|\, H)} \\ 
& = \sum_{f_u \in \calF', f_v \in \calF'} \Pro(f_u f_v) [f_u(H) = a][f_v(H) = b] \\ 
& = \sum_{f_u \in \calF', f_v \in \calF'} \Pro(f_u) \Pro(f_v) [f_u(H) = a][f_v(H) = b]  \\ 
& = \sum_{f_u \in \calF'} \Pro(f_u)[f_u(H) = a]  \cdot \sum_{f_v \in \calF'} \Pro(f_v)  [f_v(H) = b] \\ 
&= \Pro(T_u = a | H) \Pro(T_v = b | H) \ , 
\end{align*}
thereby proving independence. We have used that 
$\Pro(f_u f_v) = \Pro(f_u) \Pro(f_v)$ in the second equality, 
which follows from the fact that $f_u$ and $f_v$ are chosen \emph{a priori} 
and independently.
\qed
\end{proof}

Let $p_t$ denote the i.i.d.\ probability that some link in a given gadget 
$F_i$ transmits at time $t$.
Now $\Pro(A_i(t+1) | A_i(t)) = p_t^2 + (1-p_t)^2$, which is minimized for $p_t = \frac{1}{2}$ with value $\frac{1}{2}$. 
Thus, 
\begin{equation}
\Pro(A_i(t+1) | A_i(t)) \geq \frac{1}{2}\ .
\label{eq:a-markovian}
\end{equation}

Intuitively, on average, at most half of the active gadgets become
inactive in any given round, and thus it takes $\lg n$ rounds for all
gadgets to become inactive.

\begin{theorem}
Let $z(n)$ be a random variable whose value is the smallest time $t$
at which none of the gadgets are active.
Then, $\Ex(z(n)) = \Omega(\log n)$.
\label{lbmain}
\end{theorem}
\begin{proof}
Consider gadget $F_i$.
Note that for every $t > 0$,
$A_i(1) \cap A_i(2) \cap \cdots \cap A_i(k) = A_i(k)$
and $\Pro(A_i(0)) = 1$. Let $t_0 = \lceil \lg n \rceil$.
Then, for every $t' \ge t_0$,
\begin{align*}
\Pro(A_i(t')) & = \Pro(A_i(0))\prod_{t=2}^{t'} \Pro(A_i(t)|\cap_{j < t} A_i(j)) \\
 & = 1 \cdot \prod_{t=2}^{t'} \Pro(A_i(t)|A_i(t-1)) \\
 & \ge 2^{-(t'-1)} > \frac{1}{n}\ ,
\end{align*}
by (\ref{eq:a-markovian}).
Let $Q_{t'} = \cap_i \overline{A_i(t')}$ be the event that 
none of the $n$ gadgets are active at time $t'$.
Since events of different gadgets are independent, it holds for any $t' \ge t_0$ that
\[ \Pro(Q_{t'}) = \prod_{i=1}^n (1 - \Pro(A_i(t')) \le \left(1 - \frac{1}{n}\right)^n \le e^{-1}\ . \]
Then, by definition of expectation, 
\[ \Ex(z(n)) = \sum_{t=1}^\infty \Pr(\overline{Q_t})
\ge t_0 \cdot \Pr(\overline{Q_{t_0}}) \ge (1-e^{-1}) \lg n\ . \]
\end{proof}

Note that bounding $\Ex(z(n))$ suffices to lower bound the expected
time before all links successfully transmit, since by definition a link cannot
succeed as long as the corresponding gadget is active.

\section{Tight Bound on Analysis via $\maa$}
\label{sec:bara}

We achieved a $O(\log n)$-approximation by avoiding the measure 
$\maa$ in our analysis. In contrast, the $O(\log^2 n)$ bound in
\cite{KV10} is achieved by proving two separate bounds involving $\maa$: 
first $ALG = O(\maa \log n)$, and second $\maa = O(\chi \log n)$,
where $ALG$ is the expected time taken by the algorithm.  The
tightness of the bound on $ALG$ under any oblivious power assignment
follows from Section \ref{sec:lbound}, as it is easy to verify that
$\maa = \Theta(1)$ in that construction.  We give a construction below
for which the second bound is tight.  Thus, going through $\maa$ is
not sufficient to obtain improved bounds, and different analysis is
required.

Our construction uses uniform power. 
This is necessary, since for other oblivious power assignments $\maa = O(\chi)$, by recent results of \cite{us:SODA13}.

\begin{theorem}
For every numbers $\hat{n}$ and every number $t$, there is a set $\hat{L}$ of $\hat{n}$ links with
$\chi(\hat{L}) = \Theta(t)$ and $\maa(\hat{L}) = \Omega(\chi(\hat{L}) \log \hat{n})$ under uniform power.
\label{thm:alg-bara-lb}
\end{theorem}

This lemma shows, perhaps surprisingly, that there can be a huge difference between the in-affectance and out-affectance of a link in a feasible set, thereby illustrating the need for the bi-feasibility concept.

\begin{lemma}
For every $n$, there is a set $L$ of $n$ links on the line and a link
$l_0 \in L$, such that under uniform power, $L$ is feasible while
$a_0(L) = \Omega(\log n)$.
\label{lem:hi_out_aff}
\end{lemma}

\begin{proof}
We form the set $L = \{l_0, l_1, \ldots, l_{n-1}\}$ as follows.
The sender $s_i$ of link $l_i$ is positioned at coordinate $d(s_0, s_i) = c \cdot
i^{1/\alpha} \cdot 2^i$, where $c > 1$ is a constant to be determined.
The length of the link $l_i$ is $\ell_i = 2^i$ and the receiver $r_i$ is positioned at $r_i = s_i + \ell_i = (c\cdot i^{1/\alpha} + 1)2^i$.

Then,
\begin{align*}
a_0(L) & = \sum_{i=1}^{n-1} \left( \frac{\ell_i}{d_{0i}} \right)^\alpha  \\
 & = \sum_{i=1}^{n-1} \left( \frac{2^i}{(c\cdot i^{1/\alpha}+1) 2^i} \right)^\alpha  \\
 & < \frac{1}{(2c)^\alpha} \sum_{i=1}^{n-1} \frac{1}{i} \\
 & = \Omega(\log n)\ .
\end{align*}

To show feasibility, we first bound distances between links by:
\begin{align*}
d_{i-1,i} & = d_{0i} - d(s_0,s_{i-1}) \\
 & = (c\cdot i^{1/\alpha} + 1)2^i - c(i-1)^{1/\alpha} 2^{i-1} \\
 & > c \cdot i^{1/\alpha} 2^{i-1}\ ,
\end{align*}
and for $m > 0$,
\begin{align*}
  d_{i+m,i} & =  d(s_0, s_{i+m}) - d(s_0,r_i) \\
   & > c(i+m)^{1/\alpha} 2^{i+m} - (c\cdot i^{1/\alpha}+1)2^i \\
   & > ci^{1/\alpha} 2^{i+m} - (c\cdot i^{1/\alpha}+1)2^{i+m-1} \\
   & = (2c \cdot i^{1/\alpha} - (c\cdot i^{1/\alpha}+1))2^{i+m-1} \\
   & \ge (c-1) 2^{i+m-1} \ . 
\end{align*}
We then bound the in-affectance of each link by
\begin{align*}
 a_{L}(i) & = \sum_{k, k < i} a_k(i) + \sum_{k, k > i} a_k(i) \\
   & \le i \cdot a_{i-1}(i) + \sum_{m=1}^{n-i-1}
   \left(\frac{\ell_i}{d_{i+m,i}}\right)^\alpha \\
   & \le i \cdot \left(\frac{\ell_i}{d_{i-1,i}}\right)^\alpha
    + \sum_{m=1}^{n-i} \left(\frac{2^i}{(c-1)\cdot 2^{i+m-1}} \right)^\alpha \\
   & \le i \cdot \left(\frac{2^i}{c \cdot i^{1/\alpha}\cdot 2^{i-1}} \right)^\alpha
    + \frac{1}{(c-1)^\alpha} \sum_{m=0} \left(\frac{1}{2^{\alpha}} \right)^m \\
   & = \frac{2}{c^\alpha} + \frac{1}{(c-1)^\alpha} \cdot \frac{2^\alpha}{2^\alpha-1}\ .
\end{align*}
Thus, when $c \ge 1 + \left(3\left(1+\frac{1}{2^\alpha-1}\right)\right)^{1/\alpha}$, it holds that $a_L(i) \le 1$ for
each link $l_i$, i.e., $L$ is feasible. \qed
\end{proof}

We now turn to proving Thm.~\ref{thm:alg-bara-lb}.
We construct the set $\hat{L}$ that satisfies the
claim of the theorem. Let $L$ be the set
feasible under uniform power and
the link $l_0 = (s_0, r_0) \in L$ with $a_L(l_0) = \Omega(\log n)$,
promised by Lemma ~\ref{lem:hi_out_aff}.
Let $L_1$ denote $t$ isometric copies of $L$ with links in the same position 
as $L$.

We next take an arbitrary set $S$ on $n$ links that is feasible under uniform power, and scale its distances so that maximum pairwise distance is $o(1)$. 
For instance, we can let $S = \{l'_1, l'_2, \ldots, l'_n\}$ be $\hat{L}$ scaled by a factor of $4^{-n}$ so that the length of $l'_i=(s'_i,r'_i)$ is $2^{i-2n}$,
$s'_i = c i^{1/\alpha} 2^{i-2n}$, and $r'_i = s'_i + 2^{i-2n}$. By the same argument as Lemma 6, $S$ is feasible; observe also that pairwise distances of points within $S$ are $O(n/2^n)=o(1)$.
Let $L_2$ denote $t$ copies of $S$, with the same coordinates;
thus, the nodes of $L_2$ are all close to the node $s_0$ in $L_1$.
Finally, we form the combined instance $\hat{L} = L_1 \cup L_2$ with a total of $\hat{n} = 2t n$ links.

Observe that for every $l_j \in L_2$ and $l_i \in L_1$ ($i > 0$), that $d_{j i} = d_{0i} (1+o(1))$.
Since we use uniform power, it holds for each of the $t n$ links $l_j \in L_2$ that
\[ a_{j}(L_1) = \Theta(a_0(L_1)) = \Theta(t \cdot a_0(L)) = \Omega(t \log n)\ . \] 
Thus,
\[ a_{\hat{L}}(\hat{L}) \geq a_{L_2}(L_1) = |L_2| \Omega(t \log n)\ , \]
implying that
\[ \maa \geq \frac{1}{|\hat{L}|} a_{\hat{L}}(\hat{L}) = \Omega(t \log (\hat{n}/t))\ . \]
On the other hand, the set $\hat{L}$ clearly has a scheduling number of $2t$, as it is formed by $2t$ feasible sets.
Hence, the theorem.

\section{Conclusions}

We have given a distributed scheduling algorithm that is $O(\log
n)$-approximate in the scheduling model, and shown this factor cannot
be improved in general. Our lower bound construction, however, applies
only to instances with small scheduling number. 

A similar randomized scheduling algorithm was shown by Fangh\"anel et al.\ \cite{FKV09}
to yield an asymptotic constant-factor approximation for the case of
linear power assignment. One key difference is that in the case of linear
power, all links have low affectance ($O(\chi)$), while for general
sublinear length-monotone power assignments this only holds on average.

It remains an important and intriguing open question whether a better asymptotic
approximation ratio can be obtained.

\subsection*{Acknowledgement}
We thank Marijke Bodlaender for helpful discussions leading to the derivation of Lemma \ref{lem:hi_out_aff}.

\bibliographystyle{spmpsci}      
\bibliography{references}   

\begin{thebibliography}{10}
\providecommand{\url}[1]{{#1}}
\providecommand{\urlprefix}{URL }
\expandafter\ifx\csname urlstyle\endcsname\relax
  \providecommand{\doi}[1]{DOI~\discretionary{}{}{}#1}\else
  \providecommand{\doi}{DOI~\discretionary{}{}{}\begingroup
  \urlstyle{rm}\Url}\fi

\bibitem{AndrewsD09}
Andrews, M., Dinitz, M.: Maximizing capacity in arbitrary wireless networks in
  the {S}{I}{N}{R} model: Complexity and game theory.
\newblock In: INFOCOM, pp. 1332--1340. IEEE (2009)

\bibitem{GT2011}
\'Asgeirsson, E.I., Mitra, P.: On a game theoretic approach to capacity
  maximization in wireless networks.
\newblock In: INFOCOM (2011)

\bibitem{PODC13}
Bodlaender, M.H., Halld\'{o}rsson, M.M., Mitra, P.: {Connectivity and
  Aggregation in Multihop Wireless Networks}.
\newblock In: PODC (2013)

\bibitem{chafekar07}
Chafekar, D., Kumar, V., Marathe, M., Parthasarathy, S., Srinivasan, A.:
  {Cross-layer Latency Minimization for Wireless Networks using SINR
  Constraints}.
\newblock In: Mobihoc (2007)

\bibitem{Dinitz2010}
Dinitz, M.: Distributed algorithms for approximating wireless network capacity.
\newblock In: INFOCOM (2010)

\bibitem{FKRV09}
Fangh\"anel, A., Kesselheim, T., R\"acke, H., V\"ocking, B.: Oblivious
  interference scheduling.
\newblock In: PODC, pp. 220--229 (2009)

\bibitem{FKV09}
Fangh\"anel, A., Kesselheim, T., V\"ocking, B.: Improved algorithms for latency
  minimization in wireless networks.
\newblock In: ICALP, pp. 447--458 (2009)

\bibitem{GHW13}
Goussevskaia, O., Halld\'{o}rsson, M.M., Wattenhofer, R.: {Algorithms for
  wireless capacity}.
\newblock IEEE/ACM Transactions on Networking  (2013).
\newblock To appear

\bibitem{GHWW09}
Goussevskaia, O., Halld\'{o}rsson, M.M., Wattenhofer, R., Welzl, E.: {Capacity
  of Arbitrary Wireless Networks}.
\newblock In: INFOCOM, pp. 1872--1880 (2009)

\bibitem{gouss2007}
Goussevskaia, O., Oswald, Y.A., Wattenhofer, R.: {Complexity in Geometric
  SINR}.
\newblock In: {Mobihoc}, pp. 100--109 (2007)

\bibitem{kumar00}
Gupta, P., Kumar, P.R.: {The Capacity of Wireless Networks}.
\newblock IEEE Trans. Information Theory \textbf{46}(2), 388--404 (2000)

\bibitem{us:talg12}
Halld\'{o}rsson, M.M.: Wireless scheduling with power control.
\newblock ACM Transactions on Algorithms \textbf{9}(1), 7 (2012)

\bibitem{us:SODA13}
Halld\'orsson, M.M., Holzer, S., Mitra, P., Wattenhofer, R.: The power of
  non-uniform wireless power.
\newblock In: SODA (2013)

\bibitem{SODA11}
Halld\'{o}rsson, M.M., Mitra, P.: {Wireless Capacity with Oblivious Power in
  General Metrics}.
\newblock In: SODA (2011)

\bibitem{PODC12}
Halld\'{o}rsson, M.M., Mitra, P.: Distributed connectivity of wireless
  networks.
\newblock In: PODC (2012)

\bibitem{us:INFOCOM12}
Halld\'{o}rsson, M.M., Mitra, P.: Wireless capacity and admission control in
  cognitive radio.
\newblock In: INFOCOM, pp. 855 -- 863 (2012)

\bibitem{HW09}
Halld\'{o}rsson, M.M., Wattenhofer, R.: {Wireless Communication is in APX}.
\newblock In: ICALP, pp. 525--536 (2009)

\bibitem{KesselheimSoda11}
Kesselheim, T.: {A Constant-Factor Approximation for Wireless Capacity
  Maximization with Power Control in the {S}{I}{N}{R} Model}.
\newblock In: SODA (2011)

\bibitem{KesselheimThesis}
Kesselheim, T.: Approximation algorithms for spectrum allocation and power
  control in wireless networks.
\newblock Ph.D. thesis, RWTH Aachen University, Aachen, Germany (2012)

\bibitem{KV10}
Kesselheim, T., V\"ocking, B.: Distributed contention resolution in wireless
  networks.
\newblock In: DISC, pp. 163--178 (2010)

\bibitem{MaheshwariJD08}
Maheshwari, R., Jain, S., Das, S.R.: A measurement study of interference
  modeling and scheduling in low-power wireless networks.
\newblock In: SenSys, pp. 141--154 (2008)

\bibitem{MoWa06}
Moscibroda, T., Wattenhofer, R.: {The Complexity of Connectivity in Wireless
  Networks}.
\newblock In: INFOCOM (2006)

\bibitem{Moscibroda2006Protocol}
Moscibroda, T., Wattenhofer, R., Weber, Y.: {Protocol Design Beyond Graph-Based
  Models}.
\newblock In: {Hotnets} (2006)

\bibitem{pei2012}
Pei, G., Kumar, V.A.: Distributed link scheduling under the physical
  interference model.
\newblock In: INFOCOM (2012)

\bibitem{Tanenbaum:2002:CN:572404}
Tanenbaum, A.: Computer Networks, 4th edn.
\newblock Prentice Hall Professional Technical Reference (2002)

\bibitem{tonoyanlinear}
Tonoyan, T.: On the problem of wireless scheduling with linear power levels.
\newblock CoRR \textbf{abs/1107.4981} (2011)

\end{thebibliography}

\appendix
\section*{Affectance Reduction}

The following is given (with minor modification) in
\cite[Theorem 4.1]{GHW13}.

\begin{theorem}
\label{anti-signal-strengthening}
Let $S$ be an anti-feasible set and $p < 2$ be a value.
Then, $S$ can be partitioned into $t = \left(\frac{4}{p}\right)^2$ sets $S_1, S_2, \ldots, S_t$, each satisfying $a_v(S_i) \leq p$, for every $l_v \in S_i$.
\end{theorem}
\begin{proof}
We first partition $S$ into a sequence $T_1, T_2, \ldots$ of sets as follows. 
Order
the links in $S$ in decreasing order. For each link $l_v$, assign
$l_v$ to the first set $T_j$ for which $a_{v}(T_j) \le
p/2$, i.e.\ the accumulated affectance of $l_v$ on the
previous, longer links in $T_j$ is at most $p/2$. Since each link
$l_v$ originally had out-affectance at most $2$, then by the
additivity of affectance, the number of sets used is at most $\lceil
\frac{2}{p/2} \rceil = \lceil \frac{4}{p} \rceil$.

We then repeat the same approach on each of the sets $T_i$,
processing the links this time in increasing order.
The number of sets is again $\lceil \frac{4}{p} \rceil$ for each $T_i$, or
$\lceil \frac{4}{p} \rceil^2$ in total.
In each final slot (set), the affectance of a link on the shorter links in the
same slot is at most $p/2$. In total, then, the out-affectance of each
link is at most $2 \cdot p/2 = p$.
\end{proof}

\end{document}